\documentclass{amsart}
\newtheorem{thm}{Theorem}[section]
\newtheorem{cor}[thm]{Corollary}
\newtheorem{conj}[thm]{Conjecture}
\newtheorem{ex}[thm]{Example}
\newtheorem{lem}[thm]{Lemma}

\theoremstyle{definition}

\theoremstyle{remark}

\numberwithin{equation}{section}

\begin{document}

\title[Trace Positivity of Hurwitz Products]
{On the Positivity of the Coefficients \\ of a Certain Polynomial Defined \\ by Two Positive Definite Matrices}%
\author{Christopher J. Hillar}%
\address{Department of Mathematics,
University of California, Berkeley, CA 94720.}%
\email{chillar@math.berkeley.edu}%

\author{Charles R. Johnson}%
\address{Department of
Mathematics, College of William and Mary, Williamsburg, VA
23187-8795.} \email{crjohnso@math.wm.edu}


\begin{abstract}
It is shown that the polynomial \[p(t) = \text{Tr}[(A+tB)^m]\] has
positive coefficients when $m = 6$ and $A$ and $B$ are any two $3$-by-$3$ complex Hermitian positive definite matrices.  This case is the first that is not covered by prior,
general results.  This problem arises from
a conjecture raised by Bessis, Moussa and Villani in connection
with a long-standing problem in theoretical physics.  The full
conjecture, as shown recently by Lieb and Seiringer, is equivalent
to $p(t)$ having positive coefficients for any $m$ and any two
$n$-by-$n$ positive definite matrices.  We show that, generally,
the question in the real case reduces to that of singular $A$ and $B$, and this is
a key part of our proof.  
\end{abstract}
\maketitle

\section{Introduction}
In \cite{bessis}, while studying partition functions of quantum
mechanical systems, a conjecture was made regarding a positivity
property of traces of matrices.  If this property holds, explicit
error bounds in a sequence of Pad\'e approximants follow.
Recently, in \cite{Lieb}, and as previously communicated to us
\cite{JH}, the conjecture of \cite{bessis} was reformulated as a
question about the traces of certain sums of words in two positive
definite matrices.
\begin{conj}[BMV]\label{conjecture}
The polynomial $p(t) = \text{\rm{Tr}}\left[(A+tB)^m\right]$ has all
positive coefficients whenever $A$ and $B$ are $n$-by-$n$ positive
definite (PD) matrices.
\end{conj}
The coefficient of $t^k$ in $p(t)$ is the trace of $S_{m,k}(A,B)$,
the sum of all words of length $m$ in $A$ and $B$, in which $k$
$B$'s appear (sometimes called the $k$-th Hurwitz product of $A$
and $B$).  In \cite{JH}, among other things, it was noted that,
for $m <6$, each constituent word in $S_{m,k}(A,B)$ has positive
trace.  Thus, the above conjecture is valid for $m < 6$ and
arbitrary positive integers $n$.  It was also noted in \cite{JH}
that the conjecture is valid for arbitrary $m$ and $n < 3$.  Thus,
the first case in which prior methods do not apply and the
conjecture is in doubt, is $m = 6$ and $n = 3$.  Even in this
case, all coefficients, except
$\text{Tr}[S_{6,3}(A,B)]$, are known to be positive (also as shown
in \cite{JH}).  Our purpose here is to show that the remaining
coefficient $\text{Tr}[S_{6,3}(A,B)]$ is nonnegative when $A$ and $B$
are 3-by-3 positive definite matrices, which requires notably
different methods (some summands of $S_{6,3}(A,B)$ can have
negative trace \cite{JH}).  It follows that the conjecture is
valid for $m= 6$, $n=3$, our new result.  A key tool is that it
suffices to prove the conjecture for singular (positive
semidefinite) matrices.

The coefficients $S_{m,k}(A,B)$ may be generated via the
recurrence: \[S_{m+1,k+1}(A,B) = S_{m,k}(A,B)B + S_{m,k+1}(A,B)A\]
(variants are available).  The following lemma will be useful for
computing the $S_{m,k}$.  We give an algebraic proof although a purely combinatorial proof is also available.
\begin{lem}\label{easytracelemma}
For any two $n$-by-$n$ matrices $A$ and $B$, we have
\[\text{\rm{Tr}}\left[S_{m,k}(A,B)\right] = \frac{m}{m-k}
\text{\rm{Tr}}\left[AS_{m-1,k}(A,B)\right].\]
\end{lem}
\begin{proof}
\begin{equation*}
\begin{split}
0 = & \text{Tr}\left[ {\sum\limits_{i = 1}^m {\left( {A + tB}
\right)^{i - 1} \left( {A - A} \right)\left( {A + tB} \right)^{m -
i} } } \right] \hfill \\
= &\text{Tr}\left[ {m A\left({A + tB} \right)^{m - 1}}\right]  -
\text{Tr} \left[\sum\limits_{i = 1}^m {\left( {A + tB} \right)^{i
- 1} A\left( {A + tB} \right)^{m - i} }  \right] \hfill \\
= & \text{Tr}\left[ {mA\left( {A + tB} \right)^{m - 1} } \right] -
\left. {\text{Tr}\left[ {\frac{d} {{dy}}\left( {Ay + tB} \right)^m
} \right]\;} \right|_{y = 1}  \hfill \\
= & \text{Tr}\left[ {mA\left( {A + tB} \right)^{m - 1} } \right] -
\left. {\frac{d} {{dy}}\left[ {\text{Tr}\left( {Ay + tB} \right)^m
} \right]\;} \right|_{y = 1}. \\
\end{split}
\end{equation*}
Since $S_{m,k}(Ay,B) = y^{m-k}S_{m,k}(A,B)$, it follows that the
coefficient of $t^k$ in the last expression above is just
\[m\text{Tr}[AS_{m-1,k}(A,B)] - (m-k)\text{Tr}[S_{m,k}(A,B)],\]
which proves the lemma.
\end{proof}

\section{Reduction to the Singular Case}

Of course, when $A$ and $B$ are Hermitian, $S_{m,k}(A,B)$ is
Hermitian, but even when $A$ and $B$ are $n$-by-$n$ real symmetric PD matrices, $n > 2$, $S_{m,k}(A,B)$ need not be PD.  Examples are
easily generated, and computational experiments suggest that it is
usually not PD. We want to show that Tr[$S_{6,3}(A,B)]$ is
nonnegative for $3$-by-$3$ positive definite $A$, $B$.  This is subtle as $S_{6,3}(A,B)$ need not have positive eigenvalues, and as some words within the $S_{6,3}(A,B)$ expression can have negative trace \cite{JH}.  A main component of our argument is based on the following technical observation.

\begin{thm}\label{mainlemma}
Let $B$ be any real $n$-by-$n$ matrix, and let $A =
\text{\rm{diag}}(1,x_1,\ldots,x_{n-1})$.  Suppose that $\mathbf{a} =
(a_1,\ldots,a_{n-1}) \in [0,1]^{n-1}$, and let $D =
\text{\rm{diag}}(1,d_1,\ldots,d_{n-1})$ be such that $d_i = 0$ if $a_i
= 0$, and $d_i = 1$ otherwise.  If $\mathbf{a}$ achieves the
minimum of the function $f: [0,1]^{n-1} \to \mathbb R$ given by
$f(x_1,\ldots,x_{n-1}) = \text{\rm{Tr}}[S_{m,k}(A,B)]$, then, with $A'
= \text{\rm{diag}}(1,a_1,\ldots,a_{n-1})$, we have
\[f(a_1,\ldots,a_{n-1}) = \text{\rm{Tr}}[S_{m,k}(A',B)] = \frac{m}{m-k}
\text{\rm{Tr}}\left[DS_{m-1,k}(A',B)\right].\]
\end{thm}

\begin{proof}
Let $A'$, $B$, $D$, and $\mathbf{a} = (a_1,\ldots,a_{n-1}) \in
[0,1]^{n-1}$ be as in the hypotheses of the theorem. First suppose
that $A' = D$. Then, it is clear that the formula in the theorem
reduces to the identity in Lemma \ref{easytracelemma}.  When $A'
\neq D$, consider the differentiable function $g: [-1/2,1] \to
\mathbb R$ given by
\[g(z) = \text{Tr}\left[S_{m,k}\left(\frac{A'+zD}{1+z},B \right)\right].\]
By hypothesis, ${\bf a} \in [0,1]^{n-1}$ achieves the
minimum for $f$.  Consequently, it follows (from basic variational techniques) that
\begin{equation}\label{derivzero}
\left. {\frac{{dg(z)}} {{dz}}\;} \right|_{z = 0}  = 0.
\end{equation}
Next, notice that, \[\frac{d} {{dz}}\left[ {\text{Tr}\left(
{\frac{{A' + zD}} {{1 + z}} + tB} \right)^m } \right] =
\text{Tr}\left[ {\frac{d} {{dz}}\left( {\frac{{A' + zD}} {{1 + z}}
+ tB} \right)^m } \right]
\]
\[=  \text{Tr}\left[ {\sum\limits_{i = 1}^m {\left( {\frac{{A' + zD}}
{{1 + z}} + tB} \right)^{i - 1} \frac{d} {{dz}}\left( {\frac{{A' +
zD}} {{1 + z}} + tB} \right)\left( {\frac{{A' + zD}} {{1 + z}} +
tB} \right)^{m - i} } } \right].\]  In particular, at $z= 0$, the
above expression evaluates to
\[\text{Tr}\left[ {\sum\limits_{i = 1}^m {\left( {A' + tB} \right)^{i - 1} \left( {D - A'} \right)\left( {A' + tB} \right)^{m - i} } } \right]
\hfill\]
\[=\text{Tr}\left[ {m D\left({A' + tB} \right)^{m - 1}}\right]  - \text{Tr} \left[\sum\limits_{i = 1}^m {\left( {A' + tB} \right)^{i - 1} A'\left( {A' + tB} \right)^{m - i} }  \right]
\hfill\]
\[=\text{Tr}\left[ {mD\left( {A' + tB} \right)^{m - 1} }
\right] - \left. {\text{Tr}\left[ {\frac{d} {{dy}}\left( {A'y +
tB} \right)^m } \right]\;} \right|_{y = 1}  \hfill\]
\begin{equation}\label{traceeqn}
=\text{Tr}\left[ {mD\left( {A' + tB} \right)^{m - 1} } \right] -
\left. {\frac{d} {{dy}}\left[ {\text{Tr}\left( {A'y + tB}
\right)^m } \right]\;} \right|_{y = 1}.
\end{equation}

Finally, observe that $S_{m,k}(A'y,B) = y^{m-k}S_{m,k}(A',B)$ so
that the coefficient of $t^k$ in (\ref{traceeqn}) is
\[m\text{Tr}[DS_{m-1,k}(A',B)] - (m-k)\text{Tr}[S_{m,k}(A',B)].\]
It follows, therefore, from (\ref{derivzero}) that
\[\text{Tr}[S_{m,k}(A',B)] = \frac{m}{m-k}
\text{Tr}\left[DS_{m-1,k}(A',B)\right].\]  This completes the
proof.
\end{proof}

\begin{ex}
As an example of the theorem, let $m = 4$, $n = 3$, $k = 2$, and
\[B = \left[
\begin {array}{ccc} -2&1&0\\\noalign{\medskip}-1&2&3
\\\noalign{\medskip}1&-1&3\end {array} \right], \
A = \left[ \begin {array}{ccc}
1&0&0\\\noalign{\medskip}0&x_1&0\\\noalign{\medskip}0&0&x_2\end
{array} \right] .\]  A straightforward computation gives us that
\begin{equation*}
\begin{split}
\text{\rm{Tr}}[S_{4,2}(A,B)] = & \ 
20-4\,x_1+8\,{x_1}^{2}-12\,x_1x_2+42\,{x_2}^{2}, \\
\text{\rm{Tr}}[S_{3,2}(A,B)] = & \ 9 + 18x_2. \\
\end{split}
\end{equation*}
The minimum of $\text{\rm{Tr}}[S_{4,2}(A,B)]$ is achieved by $x_1 = 7/25$, $x_2 = 1/25$, and one has \[\text{\rm{Tr}}[S_{4,2}(A',B)] = 2\text{\rm{Tr}}[S_{3,2}(A',B)] =
\frac{486}{25}.\]
\end{ex}

Let $A$, $B$, and $f$ be as in Theorem \ref{mainlemma}. If we are
fortunate enough that $f$ achieves a minimum $f(\mathbf{a})$ with
$\mathbf{a} \in (0,1]^{n-1}$, then $D$ is the identity matrix and the theorem statement simplifies to the following.

\begin{cor}\label{maincor}
Suppose that $f$ as in Theorem \ref{mainlemma} achieves a minimum
$f(\mathbf{a})$ with $\mathbf{a} \in (0,1]^{n-1}$.  Then, the
nonnegativity of $\text{\rm{Tr}}[S_{m-1,k}(A',B)]$ implies the
nonnegativity of $\text{\rm{Tr}}[S_{m,k}(A',B)]$.
\end{cor}

To see the importance of this corollary, we next examine the real version of Conjecture \ref{conjecture}.  Suppose we know that the conjecture is true for
the power $m-1$ and also suppose (by way of contradiction) that
there exist $n$-by-$n$ real positive definite matrices $A$ and $B$
such that $\text{Tr}[S_{m,k}(A,B)]$ is negative.  Then, in
particular, (by homogeneity) there are real positive definite $A$
and $B$ with norm 1 such that $\text{Tr}[S_{m,k}(A,B)]$ is
negative (here, we use the spectral norm \cite[p. 295]{HJ1} so
that for positive semidefinite $A$, it is just the largest
eigenvalue of $A$). Let $M$ be the (compact) set of real positive
semidefinite matrices with norm 1 and choose $(A,B) \in M\times M$
that minimizes $\text{Tr}[S_{m,k}(A,B)]$; our goal is to show that
this minimum is 0. By a uniform (real) unitary similarity we may assume
that $A = \text{diag}(1,a_1,\ldots,a_{n-1})$ is diagonal with $1
\geq a_1 \geq \cdots \geq a_{n-1} \geq 0$.

Corollary \ref{maincor} then tells us that $A$ must be singular,
because by induction, $\text{Tr}[S_{m-1,k}(A,B)]$ will be
nonnegative for all positive semidefinite $A$ and $B$.  By
symmetry, it also follows that $B$ is singular.  We combine these
observations into the following theorem.

\begin{thm}\label{singthm}
Suppose that $\text{\rm{Tr}}\left[(A+tB)^{m-1}\right]$ has all positive
coefficients for each pair of $n$-by-$n$ real
positive definite matrices $A$ and $B$.  If $p(t) = \text{\rm{Tr}}\left[(A+tB)^m\right]$ has all positive coefficients whenever $A, B \neq 0$ are singular $n$-by-$n$ real
positive definite matrices, then $p(t)$ has all positive
coefficients whenever $A$ and $B$ are arbitrary $n$-by-$n$ real
positive definite matrices.
\end{thm}

\section{Symbolic Real Algebraic Geometry}

In this section, we discuss the symbolic algebra preliminaries
necessary for solving the $m = 6$, $n=3$ case of Conjecture
\ref{conjecture}. Let $R = \mathbb Q[x_1,\ldots,x_n]$, and let $I,J$ be
two ideals of $R$. The \emph{quotient ideal} of $I$ by $J$ is the
ideal of $R$ given by \cite[p. 23]{Cox2} \[(I:J) = \{f \in R :
fg \in I \text{ for all $g \in J$}\}.\]  We can iterate this
process to get the increasing sequence of ideals \[I \subseteq
(I:J) \subseteq (I:J^2) \subseteq (I:J^3) \subseteq \cdots.\] This
sequence stabilizes to an ideal called the \emph{saturation} of
$I$ with respect to $J$ (see \cite[p. 15]{PolySystems}):
\[(I,J^{\infty}) = \{f \in R :  \exists m\in \mathbb N \text{
with } f^m \cdot J \subseteq I\}.\] If $I$ is any ideal in $R$,
let $V(I)$ denote the set, \[V(I) = \{(a_1,\ldots,a_n) \in \mathbb
C^n :  f(a_1,\ldots,a_n) = 0 \text{ for all $f \in I$}\}.\]
From these definitions, it is easily verified that for any two
ideals, $I,J \subseteq R$, \[V(I) \setminus V(J) \subseteq
V(I:J^{\infty}).\]

For our particular application, we will be interested in proving
that $V(I) \setminus V(J)$ contains no elements in $(0,1)^n$.  Let
$P$ denote the saturation ideal $\left( I: J^{\infty} \right)$. If
we are fortunate enough to find that $P = \left\langle 1
\right\rangle = \mathbb Q[x_1,\ldots,x_n]$, then there are no
points in $V(I) \setminus V(J)$ (and hence none in $(0,1)^n$). One
difficulty with this approach is that these new saturations do not
always produce unit ideals.  One more idea is needed, which we
describe below.

If $K$ is an ideal of $R$, the \emph{elimination ideal} \cite[p.
25]{Cox2} of $K$ with respect to $x_i$ is $K_i = \mathbb Q[x_i]
\cap R$.  The $x_i$-coordinates of elements in $V(K)$ are elements
in $V(K_i)$.  For our purposes, we need only verify that for a
saturation $P$, there is an elimination ideal $P_i$ of $P$ such
that $V(P_i)$ contains no numbers in $(0,1)$.

Normally, a procedure such as the one outlined above would be
relatively intractable (the symbolic algorithms are doubly exponential in
nature).  Our reductions give us enough efficiency to complete a proof computationally.  We performed our computations using the symbolic algebra system Macaulay 2.  

\section{The Case $m = 6$, $n=3$}

The remainder of this article is devoted to a technical
consideration of the case $m = 6$, $k = 3$, $n=3$ which is the
content of the theorem below.

\begin{thm}
The polynomial $p(t) = \text{\rm{Tr}}[(A+tB)^m]$ has positive coefficients when $m =
6$ and $A$ and $B$ are any two $3$-by-$3$ positive definite
matrices.
\end{thm}

\begin{proof}
Suppose that there exist $3$-by-$3$ (complex Hermitian) positive definite
matrices $A$ and $B$ such that $\text{Tr}[S_{6,3}(A,B)]$ is
negative; we will derive a contradiction.  Performing a uniform unitary similarity and using homogeneity, we may assume that $A$ and $B$ are of the form, \[
A = \left[ {\begin{array}{*{20}c}
   { 1 } & 0 & 0  \\
   0 & { r } & 0  \\
   0 & 0 & {s}  \\
\end{array}} \right],\;\;B = \left[ {\begin{array}{*{20}c}
   {a} & x & z  \\
   \overline{x} & {b } & y  \\
   \overline{z} & \overline{y} & {c}  \\
\end{array}} \right],
\] in which $1 \geq r \geq s$, $a,b,c \geq 0$, and $x,y,z \in \mathbb C$.  If $x,y,z \geq 0$, then we clearly have a contradiction. Otherwise, perform a simultaneous diagonal unitary similarity on $A$ and $B$ (a similarity by a diagonal matrix with entries on the unit disc) making $x,y \geq 0$. This does not change the trace of $S_{6,3}(A,B)$.

We next show that we may assume $z \in \mathbb R$.  A computation of Tr[$S_{6,3}(A,B)$] reveals that it has the form $w = \alpha z\overline{z}+\beta z + \gamma \overline{z}+\delta$, in which $\alpha,\beta,\gamma,\delta \geq 0$.  Since $w$ is real, we have 
\begin{equation*}
\begin{split}
w = \text{Re}(w) = & \ \alpha z\overline{z}+\beta \text{Re}(z) + \gamma \text{Re}(\overline{z})+\delta \\
\geq & \ \alpha \text{Re}(z)^2+\beta \text{Re}(z) + \gamma \text{Re}(z)+\delta. \\
\end{split}
\end{equation*}
Consequently, it follows that we can assume $z$ is real and negative.  Theorem \ref{singthm} now applies, so that it is enough to verify the claim with $s = 0$ and det$(B) = 0$.  

Since $B$ is positive semidefinite, we have $ab-x^2 \geq 0$.  If $b = 0$, then $x = 0$, and an
easy computation shows that \[\text{Tr}[S_{6,3}(A,B)] =
6\,{z}^{2}c+24\,a{z}^{2}+20\,{a}^{3}+6\,{r}^{3}{y}^{2}c \geq 0,\]
a contradiction.  Therefore, we must have $b > 0$.  A similar
computation also shows that $a,x,y > 0$.

Next, we prove that $c > 0$.  Since \[\text{det}(B) =
2\,xzy+abc-a{y}^{2}-{x}^{2}c-{z}^{2}b = 0,\] it follows that when
$c = 0$, we have $2xyz = bz^2+ay^2$.  From this, it is clear
that $z < 0$ is impossible, and therefore $z = 0$, a
contradiction.  Finally, if $ab = x^2$, then from $\text{det}(B) =
0$, we have that $2xyz = bz^2+x^2y^2/b$.  This implies again that  $z
= 0$, another impossibility.  Hence, $ab-x^2 > 0$.

Summarizing these observations, we may assume that \[B = \left[
{\begin{array}{*{20}c}
   {\frac{x^2+u^2}{b}} & x & -z  \\
   x & {b } & y  \\
   -z & y & {{\frac {{x}^{2}{y}^{2}+u^2{y}^{2}+2\,xbzy+{z}^{2}{b}^{2}}{u^2b}}}  \\
\end{array}} \right]
\] in which $u,b,x,y > 0$ and $z > 0$.  Furthermore, if $r = 1$ or $r = 0$,
then \cite[Theorem 4]{HJS} (along with a straightforward
continuity argument) implies that $\text{Tr}[S_{6,3}(A,B)]$ is
nonnegative.  Therefore, we may assume that $0 < r < 1$.

A direct computation shows that $b^3u^2\text{Tr}[S_{6,3}(A,B)]$ is a polynomial $p(r,x,y,z,u,b)$ $\in \mathbb Z[r,x,y,z,u,b]$. The negative terms in $p$ factor as 
\begin{equation}\label{negterms}
-12\,{b}^{3}{u}^{2}xzy \left({r}^{2}+ r+1 \right).
\end{equation}
We shall verify that the minimum of $p(r,x,y,z,u,b)$ over
$r,x,y,z,u,b \in [0,1]$ is 0, which will prove the claim (by
homogeneity of the matrix $B$ in the variables $x,y,z,u,b$).

If any of $x,y,z,u,$ or $b$ is zero, then we are done
by (\ref{negterms}); therefore, we begin by determining the
critical points of $p$ in $(0,\infty)^6$.  This amounts to a
calculation of
\begin{equation}
D = \left\langle  \frac{{\partial p}} {{\partial
r}},\frac{{\partial p}} {{\partial
x}},\frac{{\partial p}} {{\partial y}},\frac{{\partial p}}
{{\partial z}},\frac{{\partial p}} {{\partial u}},\frac{{\partial
p}} {{\partial b}} \right\rangle,
\end{equation}
which is an ideal in the ring $\mathbb Q[r,x,y,z,u,b]$.  We are
interested in verifying that the set of points $V(D) \setminus
V(rxyzub)$ contains no element in $(0,1)^6$. From the discussion
in the previous section, it suffices to verify this claim
for $V(D: \langle rxyzub \rangle^{\infty})$.

Let $P = \left( D: \left\langle rxyzub\right\rangle^{\infty}
\right)$.  Using Macaulay 2, it can be checked that $P$ is the
unit ideal $\mathbb Q[r,x,y,z,u,b]$.  It follows that the minimum of
the function $p$ above must occur when one of the $x,y,z,u,b$ is 1
(in other words, on the ``boundary'').

This process now continues, recursively, by next finding the critical points of
the functions $p(r,1,y,z,u,b), \ldots, p(r,x,y,z,u,1)$, and checking
that they either do not occur in $(0,1)^5$ or that the function is
nonnegative when they do.  As noted before, a difficulty is that
these new saturations do not always produce unit ideals.
Therefore, we finish by showing that for each saturation $P$,
there is an elimination ideal $P_i$ of $P$ such that $V(P_i)$
contains no positive numbers in $(0,1)$. Since each $P_i$ is
generated by a single-variable polynomial, we use Sturm's
algorithm to verify such a claim symbolically.  These computations were also performed in Macaulay 2.  This completes the proof of the theorem.
\end{proof}

As a final remark, we should note that there are some good tools
for the numerical exploration of such problems. Namely, the
program SOSTOOLS written by Prajna, Papachristodoulou, and Parrilo is an excellent resource for
investigating real algebraic
systems.\footnote{http://www.cds.caltech.edu/sostools/}


\end{document}